\documentclass[pra,twocolumn,showpacs]{revtex4-1}
\usepackage{hyperref}
\usepackage{graphicx}
\usepackage{amsmath}

\newtheorem{theorem}{Theorem}

\newtheorem{corollary}[theorem]{Corollary}

\newenvironment{proof}[1][Proof]{\noindent\textbf{#1.} }{\ \rule{0.5em}{0.5em}}

\newcommand{\ket}[1]{\left\vert#1\right\rangle}
\newcommand{\bra}[1]{\left\langle#1\right\vert}

\newcommand{\beq}{\begin{equation}}
\newcommand{\eeq}{\end{equation}}
\newcommand{\baq}{\begin{eqnarray}}
\newcommand{\eaq}{\end{eqnarray}}

\newcommand{\brac}[1]{\lbrace #1 \rbrace}

\def\Re{\mathrm{Re}}
\def\ket#1{| #1 \rangle}
\def\bra#1{\langle #1 |}

\def\D{\mathcal{D}}

\def\N{\mathcal{N}}

\begin{document}
	
	\title{Capacities of Quantum Amplifier Channels}
	\author{Haoyu Qi}
	\affiliation{Hearne Institute for Theoretical Physics and  Department of Physics \& Astronomy,
	Louisiana State University, Baton Rouge, Louisiana 70803, USA}
	\author{Mark~M.~Wilde}
	\affiliation{Hearne Institute for Theoretical Physics, Department of Physics \& Astronomy,
	Center for Computation and Technology,
	Louisiana State University, Baton Rouge, Louisiana 70803, USA}
	
	\begin{abstract}
	 Quantum amplifier channels are at the core of several physical processes. Not only do they model the optical process of spontaneous parametric down-conversion, but the transformation corresponding to an amplifier channel also describes the physics of the dynamical Casimir effect in superconducting circuits, the Unruh effect, and Hawking radiation. Here we study the communication capabilities of quantum amplifier channels. Invoking a recently established minimum output-entropy theorem for single-mode phase-insensitive Gaussian channels, we determine capacities of quantum-limited amplifier channels in three different scenarios. First, we establish the capacities of quantum-limited amplifier channels for one of the most general communication tasks, characterized by the trade-off between classical communication, quantum communication, and entanglement generation or consumption.
	Second, we establish capacities of quantum-limited amplifier channels for the trade-off between public classical communication, private classical communication, and secret key generation.
		Third, we determine the capacity region for a  broadcast channel induced by  the quantum-limited amplifier channel, and we also show that a fully quantum strategy outperforms those achieved by classical coherent detection strategies.
	In all three scenarios, we find that the capacities  significantly outperform communication rates achieved with a naive time-sharing strategy.
		
	\end{abstract}

\maketitle
\section{Introduction}
Shannon laid out the foundations of classical information theory in his breakthrough 1948 paper \cite{shannon1948bell}, where he determined the ultimate communication capabilities of classical communication channels. In today's telecommunication networks, free-space or fiber-optic channels are ubiquitous. These channels use electromagnetic waves as the basis for propagating information, but the quantum-mechanical nature of the electromagnetic field demands that we should take  quantum effects into account in order to evaluate the capacities of optical channels \cite{shapiro2009quantum}. To meet this challenge, quantum Shannon theory was devised in order to determine the ultimate communication rates of quantum communication channels for various information-processing tasks (see \cite{wilde2013quantum,holevo2013quantum} for reviews). 

Much progress has been made in the study of quantum communication over bosonic continuous systems \cite{weedbrook2012gaussian}.
Bosonic Gaussian channels have an elevated status in quantum information because they model the most common noise processes in optical communication like attenuation, amplification, and phase conjugation of optical signals
\cite{weedbrook2012gaussian,holevo2013quantum}.
The quantum-limited amplifier channel
\cite{PhysRev.128.2407,PhysRevD.26.1817}
 is also a fundamental building block of any such bosonic Gaussian channel, given that it can be decomposed as the serial concatenation of a quantum-limited attenuator followed by a quantum-limited amplifier  or its phase conjugate \cite{holevo2013quantum,Giovannetti2014}. 

Interestingly, the Bogoliubov transformation governing spontaneous parametric down-conversion in a nonlinear optical system \cite{clerk2010introduction} also describes a variety of different physical processes, such as the dynamical Casimir effect \cite{moore1970quantum}, the Unruh effect \cite{unruh1976notes} and Hawking radiation \cite{hawking1972black}. For example, the gain of a quantum amplifier channel is directly related to the acceleration of an observer in the setting of the Unruh effect. By employing Einstein's equivalence principle, the Unruh effect has a correspondence in the setting of Hawking radiation, in which the amplifier gain plays the role  of the surface gravity of the black hole. For a review on the close relationship between the above phenomena, see
Ref.~\cite{nation2012colloquium}. Related, several papers have studied quantum communication in situations where relativistic effects cannot be ignored
\cite{bradler2012quantum,bradler2015black}. Thus, the importance of quantum amplifier channels in various different fields of physics suggests that studying its communication capacities has both practical and theoretical relevance.

In this paper, we first determine communication trade-offs for a quantum-limited amplifier channel in which a sender has access to the input of the amplifier and a receiver to its output. The information trade-off problem is one of the most general information-processing tasks that one can consider for a point-to-point quantum communication channel. It allows the sender and receiver to simultaneously generate or consume any of the three fundamental information resources: classical information, quantum information, and shared entanglement. The  protocol from \cite{HW08GFP,wilde2012information,wilde2012quantum} (see also \cite[Chapter~25]{W16}) establishes an achievable rate region, which yields remarkable gains over the naive strategy of time sharing, as discussed in \cite{wilde2012information,wilde2012quantum}. In this work we prove that this achievable rate region is optimal, which establishes the capacity region for this setting. In order to do so, we establish some new mathematical properties of the entropy of the bosonic thermal state (see Appendix~\ref{sec:property}), a function which is of physical interest in a variety of
contexts. We suspect that these established properties could have
application in the analysis of other communication problems and in
studies of quantum thermodynamics, but this is more appropriate to
remain as the topic of future work.

We also consider the trade-off between
public classical bits, private classical bits, and secret key bits \cite{WH10}, and we establish the capacity region in this setting as well. This capacity region clearly has relevance when using a channel for the communication of secret information in addition to ordinary, public classical information. 

Beyond the point-to-point setup, we also determine the capacity region for the single-sender,
two-receiver broadcast channel induced by a unitary dilation of the quantum-limited amplifier channel.
 We do so by first giving a rate region achieved by inputting coherent states \cite{GK04} to a noisy amplifier channel. We find that this rate region improves upon those achieved using traditional strategies such as coherent homodyne or heterodyne detection. We also prove that this rate region is optimal for quantum-limited amplifier channels by employing similar techniques that we use for the first two scenarios mentioned above. These techniques are different when compared to those used in previous works
\cite{wilde2012information,guha2007classical}
 for the setting of the pure-loss channel. 
 
 
 This paper is organized as follows. In Section~\ref{sec:conjecture}, we review the main result of \cite{PTG16}, which establishes a minimum output-entropy theorem essential for our developments here. In Section \ref{sec:trade-off}, we consider the communication  trade-off for a quantum-limited amplifier channel. After briefly reviewing the characterization of the trade-off capacity region and the achievable rate region established in \cite{wilde2012quantum}, we prove that this rate region is optimal. We then show that the trade-off capacity region outperforms that achievable with a naive time-sharing strategy. We also find that capacities decrease with increasing amplifier gain. We then consider the unitary dilation of the quantum-limited amplifier channel as a quantum broadcast channel in Section~\ref{sec:broadcast}. In the first part of Section~\ref{sec:broadcast}, we determine an achievable rate region for two receivers by using coherent-state encoding. In the second part of Section~\ref{sec:broadcast}, we prove that this achievable rate region is optimal. In the third part of Section~\ref{sec:broadcast}, we show that the capacity region outperforms those achieved by using homodyne and heterodyne detection.
 In Section~\ref{sec:public-private}, we consider the trade-off between public and private classical communication. We determine these trade-off capacities for quantum-limited amplifier channels by employing techniques similar to those from Sections~\ref{sec:trade-off} and \ref{sec:broadcast}. 
 Finally, we discuss the relationship between entropy conjectures and capacities of bosonic Gaussian channels in Section~\ref{sec:discussion}. We  conclude in Section~\ref{sec:conclusion} with a summary and some open questions.

\section{Minimum output-entropy theorem}
\label{sec:conjecture}

All of our converse proofs in this work rely on the following minimum output-entropy theorem, which holds for a single-mode, phase-insensitive quantum-limited amplifier (and its weak conjugate \cite{caruso2006degradability}) channel with a given input entropy constraint \cite{PTG16}. We restate this result as the following theorem:
\begin{theorem}[\cite{PTG16}]
	\label{th:min-out-ent-conj}
	Consider a single-mode, phase-insensitive amplifier channel $\N_{A\rightarrow B}$. Let $ H_0>0$ be a positive constant. For any input state $\rho_A$ such that $H(\rho_A)\geq H_0$, the output von Neumann entropy $H(\N_{A\rightarrow B}(\rho_A))$ is minimized when $\rho_A$ is a thermal state with mean photon number $g^{-1}(H_0)$, where
	\begin{equation}\label{eq:g}
	g(x)\equiv (x+1)\log_2(x+1)-x\log_2 x~.
	\end{equation}
	is the entropy of a thermal state with mean photon number $x$.
	The same is true for the quantum-limited weak conjugate amplifier \cite{qi2016thermal} (the complementary channel of
	$\N_{A\rightarrow B}$ \cite{caruso2006degradability}).
\end{theorem}

Theorem~\ref{th:min-out-ent-conj} provides lower bounds for certain terms in the capacity regions in \eqref{eq:trade-off-region-1st-bnd}--\eqref{eq:trade-off-region-3rd-bnd} and \eqref{eq:amplify-rate-Bob}--\eqref{eq:amplify-rate-Charlie}, which are crucial for our converse proofs. Due to additivity issues of capacity regions in quantum information theory, proofs of converses generally require a multi-mode version of the results in \cite{PTG16}. However, a quantum-limited amplifier channel, the complementary channel of which is entanglement-breaking, is a Hadamard channel \cite{Holevo2008,giovannetti2014ultimate}. It is known that the capacity regions of both the information trade-off and broadcast problems are single-letter for Hadamard channels \cite{bradler2010trade,wilde2012quantumA,WH10,wang2016capacities}.

\section{Trading quantum and classical resources}
\label{sec:trade-off}
Our first result concerns the transmission (or consumption) of classical bits, quantum bits, and shared entanglement along with the consumption of many independent uses of a quantum-limited amplifier channel. The communication  trade-off is characterized by \textit{rate triples} $(C,Q,E)$, where $C$ is the net rate of classical communication, $Q$ is the net rate of quantum communication, and $E$ is the net rate of entanglement generation. 
The triple trade-off capacity region
of a quantum channel $\N$
is the regularization of the union of regions of the following form \cite{wilde2012quantumA}
(see also \cite[Chapter~25]{W16}):
\begin{align}
C+2Q    &\leq H(  \mathcal{N}(  \rho)  )
+\sum_{x}p_X(  x)  \left[  H(  \rho
_{x})  -H(  \mathcal{N}^{c}(  \rho_{x})  )
\right]  ,\nonumber\\
Q+E    &\leq\sum_{x}p_X(  x)  \left[  H(  \mathcal{N}(
\rho_{x})  )  -H(  \mathcal{N}^{c}(  \rho_{x})
)  \right]  ,\nonumber\\
C+Q+E    &\leq H(  \mathcal{N}(  \rho)  )  -\sum
_{x}p_X(  x)  H(  \mathcal{N}^{c}(  \rho_{x})
), \label{eq:general-trade-off-region}
\end{align}
where the union is with respect to all possible input ensembles $\brac{ p_X(x), \rho_x}$
and $\rho \equiv \sum_x p_X(x) \rho_x$. Here $\N^c$ is the complementary channel of $\N$
\cite{W16}.

For a quantum-limited amplifier channel with gain parameter $\kappa \in [1,+\infty)$, the input-output transformation in the Heisenberg picture is given by the following equation:
\begin{align}
\hat{b} = \sqrt{\kappa}\hat{a}+\sqrt{\kappa-1}\hat{e}^\dagger~,
\end{align}
where $\hat{a}$, $\hat{b}$, and $\hat{e}$ are the field-mode annihilation operators corresponding to the sender's input mode, the receiver's output mode, and an environmental input in the vacuum state, respectively. In order to have a meaningful and practical communication task, we assume that the mean photon number of the input state is constrained to be no larger than $N_S \in (0,+\infty)$ for each use of the channel.

\subsection{Achievable rate region}

The achievability part of the capacity theorem for the quantum-limited amplifier channel was already
established in \cite{wilde2012information,wilde2012quantum}. The coding strategy is to employ an input ensemble of Gaussian-distributed phase-space displacements of the two-mode squeezed vacuum. We restate this result as the following theorem, which is given as Theorem~3 in \cite{wilde2012quantum}:
\begin{theorem}\label{thm:ach-region-amp}
	An achievable rate region for
	a quantum-limited amplifier channel with amplifier gain $\kappa\geq 1$ is given by the union
	of regions of the following form:
	\begin{align}
	C+ 2Q &\leq g(\lambda N_S) +g(\kappa N_S +\bar{\kappa}) 
	- g(\bar{\kappa}[\lambda N_S+1])~,\label{eq:trade-off-region-1st-bnd}\\
	Q +E & \leq g(\kappa \lambda N_S +\bar{\kappa}) - g(\bar{\kappa}[\lambda N_S+1])~,\label{eq:trade-off-region-2nd-bnd}\\
	C+Q+E & \leq g(\kappa N_S+\bar{\kappa}) - g(\bar{\kappa}[\lambda N_S+1])~\label{eq:trade-off-region-3rd-bnd},
	\end{align} 
	where $\lambda\in [0,1]$ is a photon-number-sharing parameter and g(x) is defined in \eqref{eq:g}. The parameter $\bar{\kappa} \equiv \kappa-1$ denotes the mean number of photons generated by the channel when  the vacuum is input. 
\end{theorem}

\subsection{Outer bound for the capacity region}

Our contribution here is to prove that the rate region in Theorem~\ref{thm:ach-region-amp} is equal to the capacity region.
\begin{theorem}\label{thm:main-result}
	The triple trade-off capacity region for
	a quantum-limited amplifier channel with amplifier gain $\kappa\geq 1$ is equal to the rate region given in Theorem~\ref{thm:ach-region-amp}.
\end{theorem}

\begin{proof}
	We first recall that the capacity region of a quantum limited amplifier channel is single-letter \cite{bradler2010trade} due to the fact that a quantum-limited amplifier channel is a Hadamard channel \cite{Holevo2008,giovannetti2014ultimate}. Thus, there is no need to consider the regularization of \eqref{eq:general-trade-off-region}. To give an upper bound on the single-letter capacity region of the quantum-limited amplifier channel, we prove that for all input ensembles
	$\brac{ p_X(x), \rho_x}$ there exists a $\lambda \in [0,1]$ such that 
	the following four inequalities hold
	\begin{align}
	H(\N(\rho))&\leq g(\kappa N_S +\kappa -1)
	\label{eq:amp-ineq-1}~,\\
	\sum_x p_X(x)H(\rho_x)&\leq g(\lambda N_S)
	\label{eq:amp-ineq-2}~,\\
	\sum_x p_X(x)H(\N(\rho_x)) &\leq
	g(\kappa\lambda N_S +\kappa -1)~,\label{eq:amp-ineq-3}\\
	\sum_x p_X(x)H(\N^{c}(\rho_x)) &\geq
	g((\kappa-1)(\lambda N_S+1))~. \label{eq:amp-ineq-4}
	\end{align}

	We start by establishing the inequality in \eqref{eq:amp-ineq-1}:
	\begin{equation}
	H(\N(\rho)) 
	\leq  g(\kappa N_S + \kappa -1)~.
	\label{eq:first-ineq-1}
	\end{equation}
	This inequality follows from the facts that the output state has mean photon number no larger than $\kappa N_S +\kappa -1$ when the input mean photon number is no larger than $ N_S$ and because the thermal state of mean photon number $\kappa N_S +\kappa -1$ realizes the maximum entropy at the output.
	
	We now argue the inequalities in \eqref{eq:amp-ineq-2}
	and \eqref{eq:amp-ineq-3}. Consider that concavity of entropy and that the thermal state realizes the maximum entropy imply the following bound:
	\begin{align}
	\sum_x p_X(x) H(\rho_x)\leq H(\rho)\leq g(N_S)~.
	\end{align}
	Since $g(x)$ is monotonically increasing, there exists a $\lambda' \in [0,1]$ such that
	\begin{align}
	\sum_xp_X(x)H(\rho_x)= g(\lambda' N_S)~. \label{eq:input-ent-ineq}
	\end{align}
	From concavity of entropy and \eqref{eq:first-ineq-1}, we find that
	\begin{align}
	\sum_xp_X(x)H(\N(\rho_x)) &\leq H(\N(\rho)) \\
	&\leq g(\kappa N_S+\kappa - 1)~.
	\end{align}
	Due to the fact that the vacuum-state input realizes the minimum output entropy for any phase-insensitive quantum Gaussian channel \cite{giovannetti2014ultimate}, the following lower bound applies
	\begin{align}
	H(\N(\rho_x)) \geq g(\kappa-1)~. \label{eq:min-out-ent-vac-in}
	\end{align}
	Since $g(x)$ is monotonically increasing and since we have shown that
	\begin{equation}
	g(\kappa-1) \leq \sum_xp_X(x)H(\N(\rho_x))  \leq g(\kappa N_S+\kappa - 1),
	\end{equation}
	there exists $\lambda\in[0,1]$ such that
	\begin{align}\label{eq:output-ent-ineq}
	\sum_xp_X(x)H(\N(\rho_x)) = g(\lambda\kappa N_S +\kappa-1)~.
	\end{align}
	However, $\lambda$ and $\lambda'$ are different in general. But we can use Theorem~\ref{thm:second-prop-gx} in Appendix \ref{sec:property} to establish that $\lambda'\leq\lambda$. 
	To use it we need to know the entropy of the input state. Supposing that the mean photon number of $\rho_x$ is $N_{S,x}$, we have that 
	\begin{equation}
	H(\rho_x)  
	\leq  g(N_{S,x})~.
	\end{equation}
	Therefore there exists $\lambda'_x \in [0,1]$ such that
	\begin{align}
	H(\rho_x) = g(\lambda_x' N_{S,x})~.
	\end{align}
	Now employing Theorem~\ref{th:min-out-ent-conj} for the quantum-limited amplifier channel, we have that
	\begin{align}
	H(\N(\rho_x))\geq g(\kappa\lambda'_x N_{S,x} +\kappa -1)~,
	\end{align}
	which in turn implies that
	\begin{align}
	&\!\!\!\!g(\lambda \kappa N_S +\kappa-1)\nonumber\\&=\sum_xp_X(x)H(\N(\rho_x))\\
	&\geq \sum_xp_X(x)g(\kappa\lambda'_x N_{S,x} +\kappa -1)~.
	\label{eq:lower-bnd-HBgivX}
	\end{align}
	Together with $\sum_xp_X(x)g(\lambda'_x N_{S,x}) = g(\lambda' N_S)$, and using Theorem~\ref{thm:second-prop-gx} in Appendix~\ref{sec:property} with $q = \kappa$  we find that
	\begin{multline}
	\sum_xp_X(x) g(\kappa\lambda'_x N_{S,x} +\kappa -1)\\
	\geq g(\kappa \lambda' N_S +\kappa-1)~,
	\label{lower-bnd-HBgivX-2}
	\end{multline}
	which, by combining \eqref{eq:lower-bnd-HBgivX}
	and \eqref{lower-bnd-HBgivX-2}, implies that
	\begin{align}
	g(\lambda\kappa N_S +\kappa-1)\geq g(\kappa \lambda' N_S +\kappa-1)~.
	\end{align}
	Since $g$ is monotonically increasing and it has a well-defined inverse function, we find that
	\begin{align}
	\lambda \geq \lambda'~,
	\end{align}
	which, after combining with \eqref{eq:input-ent-ineq} and the monotonicity of $g(x)$, implies that
	\begin{align}
	\sum_xp_X(x)H(\rho_x)\leq g(\lambda N_S)~.
	\end{align}
	This concludes the proof of the inequalities  in \eqref{eq:amp-ineq-2}
	and~\eqref{eq:amp-ineq-3}.
	
	To prove the last bound in \eqref{eq:amp-ineq-4}, by \eqref{eq:min-out-ent-vac-in} and	
	\begin{equation}
	H(\N(\rho_x)) 
	\leq g(\kappa N_{S,x}+\kappa-1)~,
	\end{equation}
	we can conclude that there exists $\lambda_x \in [0,1]$ such that the following equality holds
	\begin{align}\label{eq:out-ent}
	H(\N(\rho_x)) = g(\lambda_x\kappa N_{S,x} +\kappa-1)~.
	\end{align}
	
	The quantum-limited amplifier
	channel $\mathcal{N}$ is degradable \cite{caruso2006degradability},
	and
	its degrading channel $\D_{B\rightarrow C}$  is the weakly-conjugate channel of the quantum-limited amplifier with $\kappa' = (2\kappa -1)/\kappa$ \cite{caruso2006degradability}. The main property of this degrading channel that we need is that an input thermal state
	of mean photon number $K$ leads to an output
	thermal state of mean photon number
	$(\kappa'-1)(K+1)$.
	Theorem~\ref{th:min-out-ent-conj} applied to this case
	gives that for given input entropy $g(K)$, the minimum output entropy of $\D_{B\rightarrow C}$ is
	equal to $g((\kappa'-1)(K+1))$.
	By applying it, we find that
	\begin{align}
	\nonumber
	& \!\!\!\!\sum_x p_X(x)H(\N^{c}(\rho_x)) \nonumber\\
	&\geq \sum_xp_X(x)g((\kappa'-1)(\lambda_x\kappa N_{S,x} +\kappa))\\
	&=\sum_xp_X(x) g((\kappa-1)\lambda_x N_{S,x} + \kappa-1).
	\end{align}
	Since $\sum_xp_X(x)g(\kappa\lambda_x N_{S,x} +\kappa-1) = g(\lambda\kappa N_S+\kappa-1)$, using 
	Theorem~\ref{thm:guha-A-3}
	in Appendix \ref{sec:property} with $q = (\kappa -1)/\kappa$ and $C=(\kappa -1)/\kappa$, we find that
	\begin{align}
	\nonumber
	& \!\!\!\!\sum_x p_X(x)H(\N^{c}(\rho_x)) \nonumber\\
	&\geq g(q(\lambda\kappa N_S+\kappa-1)+(\kappa-1)/\kappa)~,\\
	&=g((\lambda N_S+1)(\kappa-1))~.
	\end{align}
	This concludes our proof for the four bounds in \eqref{eq:amp-ineq-1}--\eqref{eq:amp-ineq-4}. Together with the achievability part in \cite{wilde2012quantum} (recalled as Theorem~\ref{thm:ach-region-amp}), this concludes the proof that the union of regions given by \eqref{eq:trade-off-region-1st-bnd}--\eqref{eq:trade-off-region-3rd-bnd} is equal to the quantum dynamic capacity region for
	the quantum-limited amplifier channel.
\end{proof}

\bigskip
Returning to our discussion from the introduction, we note that Theorem~\ref{thm:main-result} completely characterizes the communication abilities of any phase-insensitive quantum-limited amplifier channel, particular examples of this channel occurring in a number of scenarios of physical interest, including spontaneous parametric down-conversion in a nonlinear optical system \cite{clerk2010introduction}, the dynamical Casimir effect \cite{moore1970quantum}, the Unruh effect \cite{unruh1976notes}, and Hawking radiation \cite{hawking1972black}. That is, if one desires to use any such channel for sending classical and quantum information along with the assistance of shared entanglement, then Theorem~\ref{thm:main-result} sets the ultimate limits for such a task. Theorem~\ref{thm:main-result} thus subsumes and places a capstone on much previous literature in quantum information having to do with capacities of phase-insensitive, quantum-limited amplifier channels.

\subsection{Comparison with time-sharing strategy and large $\kappa$ limit}
\begin{figure}
	\centering
	\includegraphics[width=\columnwidth]{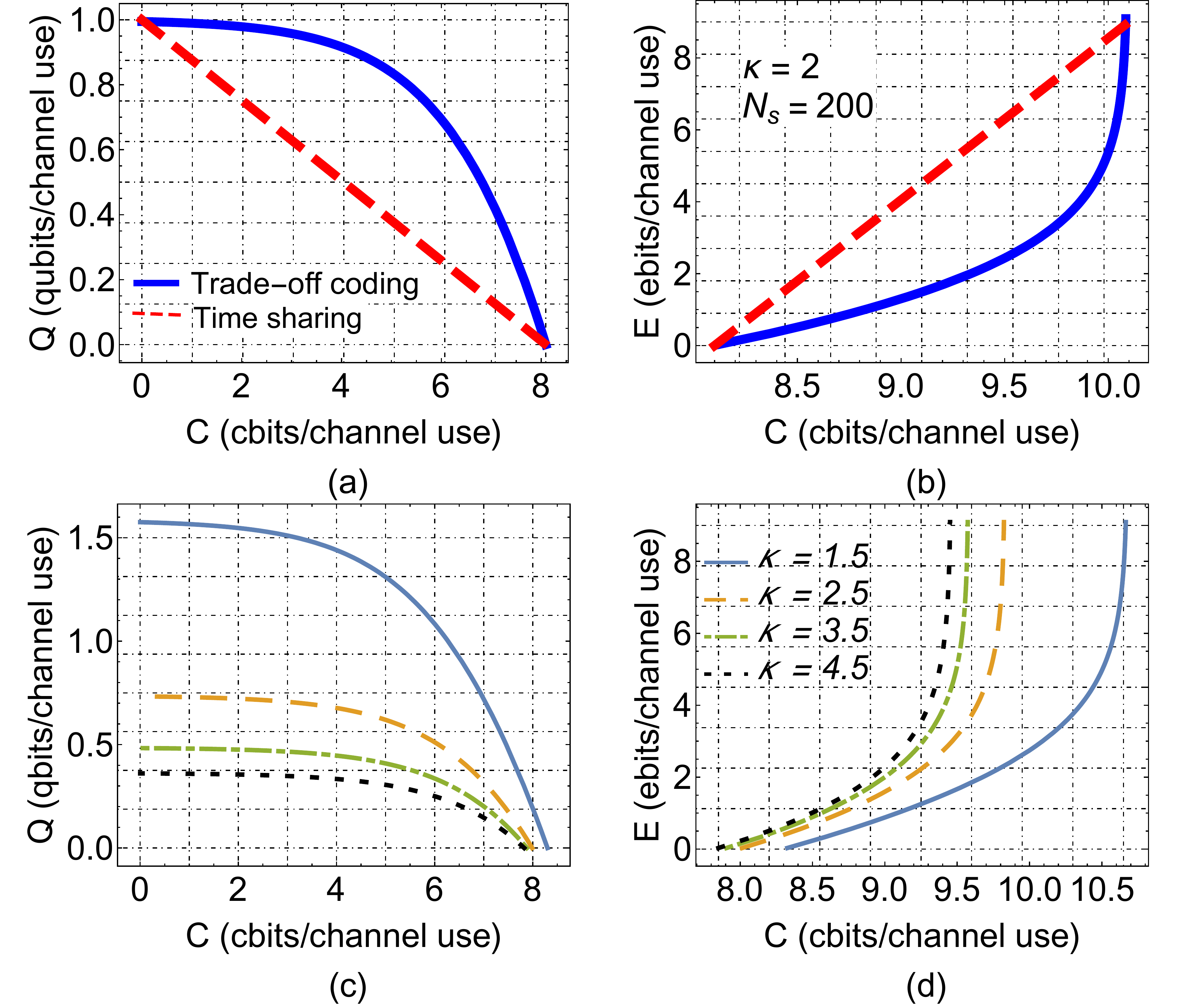}
	\caption{We consider a quantum-limited amplifier channel with $\kappa = 2$ and mean photon number constraint $N_S=200$. In (a), we plot the $(C,Q)$ trade-off. The maximum quantum capacity is equal to $\log_2(2)-\log_2(1)=1$ qubit per channel use. A trade-off coding strategy shows an improvement compared to time sharing, wherein we see that the classical data rate can be boosted while still maintaining a high quantum transmission rate.  In (b) we plot the $(C,E)$ trade-off. The sender and the receiver share entanglement, and the sender would like to transmit classical information while minimizing the consumption of entanglement. As can be seen, with trade-off coding, the sender can significantly reduce the consumption of entanglement while still keeping the classical communication rate near to its maximum value. 		
	In (c) and (d) we plot the capacity region for
	the $(C,Q)$ and $(C,E)$ trade-off with amplifier gain $\kappa = 1.5$, $2.5$, $3.5$, and $4.5$. Each capacity region shrinks as the amplifier gain $\kappa$ increases.}
	\label{fig:bosonic-trade-offs}
\end{figure}

Figure~\ref{fig:bosonic-trade-offs} displays two  special cases of the capacity region in \eqref{eq:trade-off-region-1st-bnd}--\eqref{eq:trade-off-region-3rd-bnd}. We consider a quantum-limited amplifier channel with gain $\kappa = 2$ and choose the mean  input photon number to be $N_S =200$. In Figure~\ref{fig:bosonic-trade-offs}(a), we plot the trade-off between classical and quantum communication without entanglement assistance. The maximum quantum transmission rate is $\log_2(\kappa/\bar{\kappa}) = 1$ qubits per channel use, established jointly in \cite{wolf2007quantum,PLOB15} (see discussion in \cite{WQ16}). This result also follows from the results of the present paper by considering that the bound in \eqref{eq:trade-off-region-2nd-bnd} for $\lambda=1$ gives the finite-energy quantum capacity of the quantum limited amplifier channel:
\begin{equation}
g(\kappa  N_S +\bar{\kappa}) - g(\bar{\kappa}[ N_S+1]).
\end{equation}
Taking the infinite-energy limit, we recover the formula established in \cite{wolf2007quantum,PLOB15}:
\begin{equation}
\lim_{N_S\to \infty} g(\kappa  N_S +\bar{\kappa}) - g(\bar{\kappa}[ N_S+1])= \log_2(\kappa/\bar{\kappa}).
\end{equation}
 Around 200 photons per channel use is large enough to approximate this quantum capacity well for the above parameter choices. The figure indicates a remarkable improvement over a time-sharing strategy, in which the sender transmits classical information for some fraction of the channel uses and transmits quantum information for the other fraction. By using a trade-off coding strategy, lowering the quantum data rate by about 0.1 qubits per channel use allows for sending roughly three extra classical bits per channel use. However, if a time-sharing strategy is adopted, lowering the quantum data rate by the same amount gives only one additional bit per channel use. 

In Figure~\ref{fig:bosonic-trade-offs}(b), we plot the trade-off between entanglement-assisted and unassisted classical communication. Again, a trade-off coding strategy gives a dramatic improvement over time sharing. In this figure, we take the convention that positive $E$ corresponds to entanglement consumption. With mean photon number $N_S = 200$, the sender can reliably transmit a maximum of around 10.2 classical bits per channel use by consuming around 9.1 entangled bits per channel use \cite{bennett2002entanglement,giovannetti2003entanglement}. By using trade-off coding, the sender can  reduce the consumption of entanglement to around 4 entangled bits per channel use, while still being able to transmit classical data at around 9.8 bits per channel use. 

One trend we see for the quantum-limited amplifier channel is that a large amplifier gain $\kappa$ compromises its  communication ability, as shown in Figures
\ref{fig:bosonic-trade-offs}(c) and
\ref{fig:bosonic-trade-offs}(d). For the $(C,Q)$ trade-off, as $\kappa$ increases, the quantum capacity decreases for a fixed classical rate. For the $(C,E)$ trade-off, not only the maximum classical rate is reduced, but the savings of entangled bits for a constant classical rate are also diminished. This effect results from the fact that a quantum-limited amplifier channel with large~$\kappa$ generates more photons from the vacuum, and thus injects more noise into the transmitted quantum signal. Mathematically the shrinkage of the capacity region is due to the term
$g(\bar{\kappa}[\lambda N_S+1])$
appearing in all of the inequalities in \eqref{eq:trade-off-region-1st-bnd}--\eqref{eq:trade-off-region-3rd-bnd}, which increases with increasing amplifier gain.

\section{Quantum broadcast amplifier channel}
\label{sec:broadcast}
Our next result concerns the classical capacity of a quantum broadcast channel induced by a unitary dilation of the quantum amplifier channel. We consider the
single-sender, two-receiver case in which Alice simultaneously transmits classical data to Bob ($B$) via the amplifier channel and to Charlie ($C$) via its complementary channel. The full Bogoliubov transformation for this setup is given by
\begin{align}\label{eq:amplifier-broadcast}
\nonumber
\hat{b} = \sqrt{\kappa}\hat{a} + \sqrt{\kappa-1}\hat{e}^\dagger ~,  \\
\hat{c}^\dagger = \sqrt{\kappa-1}\hat{a} + \sqrt{\kappa}\hat{e}^\dagger ~, 
\end{align}
where $\hat{a}$, $\hat{b}$, $\hat{c}$, and $\hat{e}$ are the field-mode annihilation operators corresponding to the sender Alice's input mode, the receiver Bob's output mode, the receiver Charlie's output mode, and an environmental input, respectively.
Here we consider a general amplifier channel with thermal noise, in which the input state represented by $\hat{e}$ is a thermal state with mean photon number $N_B$.
Such a channel could model information propagation
to two observers, one outside and one beyond the event horizon of a black hole \cite{bradler2015black}.
This channel could also model information propagation from an inertial observer to two constantly accelerated complementary observers moving with opposite accelerations in two causally disconnected regions of Rindler spacetime, if we take the convention that the inertial observer can encode information into Unruh modes, which arguably allows for computing estimates for an upper bound of channel capacities between inertial and relativistically accelerating observers \cite{PhysRevA.86.062307}.

 The classical capacity region of the two-user degraded quantum broadcast channel was derived in \cite{yard2011quantum} (see also \cite{SW13} for the achievability part) and found to be equal to the regularization of the union of the following rate regions:
 \begin{align}
 R_B &\leq \sum_x p_X(x)\big[H(\N(\rho_{x}))-\sum_y p_{Y|X}(y|x)H(\N(\rho_y)\big)]~,\nonumber\\ 
 R_C & \leq H({\N^c}(\rho))-\sum_xp_X(x)H(\N^c(\rho_{x}))~,
 \end{align}
 where the union is with respect to input ensembles $\brac{p_X(x)p_{Y|X}(y|x),\rho_{y}}$ with
 \begin{align}
 \rho_{x} & \equiv \sum_y p_{Y|X}(y|x) \rho_y ~,\\
 \rho & \equiv \sum_x p_X(x) \rho_x~.
 \end{align}
 
In the following we first give an achievable rate region for an amplifier channel with thermal noise. We then prove that this rate region is optimal if the multi-mode version of Theorem~\ref{th:min-out-ent-conj} is true. For the case in which the amplifier channel is quantum-limited, the capacity region is single-letter \cite{wang2016capacities} and therefore Theorem~\ref{th:min-out-ent-conj} implies the broadcast capacity region for the quantum-limited amplifier channel.

\subsection{Achievable rate region by coherent-state encoding}
\begin{theorem}
	Consider a quantum broadcast amplifier channel
	as given in \eqref{eq:amplifier-broadcast} with amplifier gain $\kappa \geq 1$ and environmental thermal-state input with mean photon number $N_B$. Suppose that the mean input photon number for each channel use is no larger than $N_S$. Then the following rate region for Bob and Charlie
	\begin{align}
	R_B &\leq g(\kappa\lambda N_S +\bar{\kappa}(N_B+1)) - g(\bar{\kappa}(N_B+1))~,\label{eq:amplify-rate-Bob}\\
	R_C &\leq g(\bar{\kappa}(N_S+1)+\kappa N_B) - g(\bar{\kappa}(\lambda N_S+1)+\kappa N_B)~,
	\label{eq:amplify-rate-Charlie}
	\end{align}	
	with $\lambda \in [0,1]$ is achievable by using coherent-state encoding
	according to the following ensemble:
	\begin{align}\label{eq:input-distribution}
	\brac{p(t) p(\alpha|t), \ket{\alpha}\bra{\alpha}}~,
	\end{align}
	where
	\begin{align}
	p(t) &= \frac{1}{\pi N_S}\exp\!\left(-\frac{|t|^2}{N_S}\right)~,\label{eq:broad-dist-1}\\
	p(\alpha|t) &= \frac{1}{\pi \lambda N_S}\exp\!\left(-\frac{|\sqrt{1-\lambda}t-\alpha|^2}{\lambda N_S}\right)~.\label{eq:broad-dist-2}
	\end{align}
	Here $\alpha$ and $t$ are complex variables and $\bar{\lambda} = 1-\lambda$.
\end{theorem}

\begin{proof}
	 Using \eqref{eq:broad-dist-1} and \eqref{eq:broad-dist-2}, we find that
	 \begin{align}
	 \nonumber
	 \rho_t &= \int d^2\alpha \ p(\alpha|t)\ \ket{\alpha}\bra{\alpha}\\
	 \nonumber
	 &= \int d^2\gamma\
	 \frac{1}{\pi \lambda N_S}\
	 \exp\!\left(-\frac{|\gamma|^2}{N_S\lambda}\right)\
	 \ket{\gamma +\sqrt{\bar{\lambda}}t}\bra{\gamma +\sqrt{\bar{\lambda}}t}\\
	 &= D(\sqrt{\bar{\lambda}}t)\rho^{\operatorname{th}}_{\lambda N_S}D^\dagger(\sqrt{\bar{\lambda}}t)~.
	 \end{align}
	 	 In the above, $D(\alpha)$ is a displacement operator \cite{GK04} and $\rho^{\operatorname{th}}_{\lambda N_S}$ denotes a thermal state of mean photon number $\lambda N_S$.
	 The overall average input state is
	 \begin{align}
	 \nonumber
	 \rho &= \int d^2t\ p(t)\ \rho_t~,\\
	 \nonumber
	 &= \int d^2t'\ \frac{1}{\pi\bar{\lambda}N_S}
	 \
	 \exp\!\left(-\frac{|t'|^2}{\bar{\lambda} N_S}\right)D(t')\rho^{\operatorname{th}}_{\lambda N_S}D^\dagger(t')\\
	 &= \rho^{\operatorname{th}}_{N_S}~,\label{eq:input-thermal}
	 \end{align}
	 which is just a thermal state with mean photon number~$N_S$, in agreement with the energy constraint. There are four entropies we need to evaluate in \eqref{eq:amplify-rate-Bob} and \eqref{eq:amplify-rate-Charlie}. The first one is
	 \begin{align}
	 \nonumber
	 &\int d^2t \ p(t) \ H(\N(\rho_t))\\\nonumber
	 &= \int d^2t \ p(t) \ H(\N (D(\sqrt{\bar{\lambda}}t)\rho^{\operatorname{th}}_{\lambda N_S}D^\dagger(\sqrt{\bar{\lambda}}t)))\\
	 \nonumber
	 &=\int d^2t \ p(t) \ H(\N(\rho^{\operatorname{th}}_{\lambda N_S}))
	 = H(\N(\rho^{\operatorname{th}}_{\lambda N_S}))\\
	 &= g(\kappa\lambda N_S +(\kappa-1)(N_B+1))~.
	 \end{align}
	 The second equality follows because the amplifier channel is covariant with respect to displacement operators and the fact that entropy is invariant with respect to a unitary transformation. 
	 
	 Since the output state is unitarily related to a thermal state with mean photon number $(\kappa-1)(N_B+1)$ when Alice sends a coherent state into an amplifier channel,
	 the second term in \eqref{eq:amplify-rate-Bob} is given by
	 \begin{multline}
	 \int d^2t \ d^2\alpha\  p(t)\ p(\alpha|t)\ H(\N(\ket{\alpha}\bra{\alpha}))
	 \\=g((\kappa-1)(N_B+1)) ~.
	 \end{multline}
	 Now similarly for (\ref{eq:amplify-rate-Charlie}), the first term is 
	 \begin{align}
	 H(\N^c(\rho^{\operatorname{th}}_{N_S})) = g((\kappa-1)(N_S+1)+\kappa N_B)~.
	 \end{align}
	 The last term can be calculated as follows:
	 \begin{align}\nonumber
	 &\int d^2t \ p(t)\ H(\N^c(\rho_t))\\\nonumber
	  &= \int d^2t \ p(t)\   H(\N^c (D(\sqrt{\bar{\lambda}}t)\rho^{\operatorname{th}}_{\lambda N_S}D^\dagger(\sqrt{\bar{\lambda}}t)))\\
	 \nonumber
	 &=\int  d^2t\ p(t)\  H(\N^c(\rho^{\operatorname{th}}_{\lambda N_S}))\\
	 &= g((\kappa-1)(\lambda N_S+1)+\kappa N_B)~.
	 \end{align}
	 We use the facts that a gauge-contravariant bosonic Gaussian channel is contravariant with respect to displacement operators and that entropy is invariant with respect to a  unitary transformation. Combining the above results, we conclude that the rate region in \eqref{eq:amplify-rate-Bob} and \eqref{eq:amplify-rate-Charlie} is achievable.
\end{proof}

\subsection{Outer bound for the capacity region}

We first prove that the rate region
in \eqref{eq:amplify-rate-Bob} and
\eqref{eq:amplify-rate-Charlie} is optimal if a multi-mode version of Theorem~\ref{th:min-out-ent-conj} is true. To do so we need to show that it is also an outer bound for the capacity region. 
\begin{theorem}
	Consider a quantum amplifier channel with amplifier gain $\kappa \geq 1$ and environmental thermal-state input with mean photon number $N_B$. Suppose that the mean input photon number for each channel use is no larger than $N_S$. Suppose that a multi-mode version of Theorem~\ref{th:min-out-ent-conj} is true. Then the region given by \eqref{eq:amplify-rate-Bob} and \eqref{eq:amplify-rate-Charlie} is an outer bound for the broadcast capacity region.
\end{theorem}
\begin{proof}
Since a general quantum amplifier channel with thermal noise is not a Hadamard channel, we need to consider the $n$-letter version of \eqref{eq:amplify-rate-Bob} and \eqref{eq:amplify-rate-Charlie}. Specifically, we need to prove
that for all input ensembles
$\brac{p_X(x)p_{Y|X}(y|x),\rho_{y}}$ for $n$ uses of the channel, there exists $\lambda \in [0,1]$ such that 
the following four bounds hold
\begin{align}
\sum_x p_X(x)H(\N^{\otimes n}(\rho_x)) \leq n g(\kappa\lambda N_S +\bar{\kappa}(N_B+1)) ~,\\
H((\N^c)^{\otimes n}(\rho)) \leq ng(\bar{\kappa}(N_S+1) +\kappa N_B)~,\\
\sum_x\sum_y p_X(x)p_{Y|X}(y|x)H(\N^{\otimes n}(\rho_y)) \geq ng(\bar{\kappa}(N_B+1))~,\\
\sum_xp_X(x)H((\N^c)^{\otimes n}(\rho_x)) \geq ng(\bar{\kappa}(\lambda N_S+1)+\kappa N_B) ~.
\end{align}

The second inequality holds because 
\begin{align}
H((\N^c)^{\otimes n}(\rho)) & \leq \sum_{j=1}^n H(\rho^j_C)
\\
& \leq ng((\kappa-1)(N_S+1) +\kappa N_B)~.
\end{align}
The first inequality follows from the subadditivity of quantum entropy. The second inequality follows from the fact that each output state at $C$ has mean photon number $(\kappa-1)(N_S+1) +\kappa N_B$ and the thermal state maximizes the entropy.

Since the vacuum minimizes the output entropy for any phase-insensitive Gaussian channel \cite{giovannetti2014ultimate}, we find that $H(\N^{\otimes n}(\rho_y))\geq ng((\kappa-1)(N_B+1))$, which leads to the third bound:
\begin{multline}
\sum_x\sum_y p_X(x)p_{Y|X}(y|x)H(\N^{\otimes n}(\rho_y)) \\ \geq ng((\kappa-1)(N_B+1))~.
\end{multline}
Now we prove the first bound. From the concavity of quantum entropy, we have that
\begin{multline}\nonumber
H\!\left(\sum_y p_{Y|X}(y|x)\N^{\otimes n}(\rho_y)\right)\\ \geq \sum_y p_{Y|X}(y|x) H(\N^{\otimes n}(\rho_y))~.
\end{multline}
Thus we have
\begin{align}\nonumber
&\sum_x p_X(x)H(\N^{\otimes n}(\rho_x)) \\\nonumber
&\geq\sum_{x,y}p_X(x)p_{Y|X}(y|x) H(\N^{\otimes n}(\rho_y))\\
&\geq ng((\kappa-1)(N_B+1))~.\label{eq:lowerbound}
\end{align}
On the other hand, we have that
\begin{align}\nonumber
\sum_x p_X(x)H(\N^{\otimes n}(\rho_x)) &\leq H(\N^{\otimes n}(\rho))\\
&\leq ng(\kappa N_S +(\kappa-1)(N_B+1))~.
\end{align}
Together with \eqref{eq:lowerbound} and the fact that $g(x)$ is monotonic, there exists $\lambda \in [0,1]$ such that
\begin{align}\nonumber
\sum_x p_X(x)H(\N^{\otimes n}(\rho_x)) = n g(\kappa\lambda N_S +(\kappa-1)(N_B+1))~.
\end{align}
To prove the last bound, we use the fact that the weakly degrading channel of the amplifier channel is the weakly-conjugate of an amplifier channel with $\kappa' = (2\kappa-1)/\kappa >1$ \cite{caruso2006degradability}. We first calculate the entropy of the output state:
\begin{align}
\nonumber
H(\N^{\otimes n}(\rho_x)) &= H(\rho_{B,x})~, \\
\nonumber
&\leq \sum_{j=1}^nH(\rho^j_{B,x})~,\\
\nonumber
&\leq n\sum_{j=1}^n\frac{1}{n}g(\kappa N_{S,x_j}+(\kappa-1)(N_B+1))~,\\
&\leq ng(\kappa N_{S,x} +(\kappa-1)(N_B+1))~.
\end{align}
The first inequality follows from subadditivity of quantum entropy. Letting $N_{S,x_j}$ be the mean photon number for the $j$th symbol of $\rho_x$, the second inequality follows because the thermal state maximizes the entropy. Letting $N_{S,x} = \sum_j N_{S,x_j}/n$, the last inequality follows from concavity of $g(x)$. Since we also have that
\begin{equation}
H(\N^{\otimes n}(\rho_x))\geq ng((\kappa-1)(N_B+1)),
\end{equation}
there exists $\lambda_x\in [0,1]$ such that
\begin{align}
H(\N^{\otimes n}(\rho_x)) = ng(\kappa\lambda_xN_{S,x} +(\kappa-1)(N_B+1))~.
\end{align}
Using the multi-mode version of Theorem~1 for the degrading channel, we find that
\begin{align}
\nonumber
&\sum_xp_X(x)H((\N^c)^{\otimes n}(\rho_x)) \\\nonumber
\geq& \sum_xp_X(x)ng((\kappa'-1)[\kappa\lambda_xN_{S,x} +\bar{\kappa}(N_B+1)+1] +\kappa'N_B)\\
=&\sum_xp_X(x)ng((\kappa-1)(\lambda_xN_{S,x}+1) +\kappa N_B)~.
\end{align}
Together with 
\begin{multline}
\sum_xp_X(x)g(\kappa\lambda_x N_{S,x}+\bar{\kappa}(N_B+1)) \\ 
= g(\kappa\lambda N_S+\bar{\kappa}(N_B+1))~,
\end{multline}
we can invoke Theorem~\ref{thm:guha-A-3} in
Appendix~\ref{sec:property} with $q = (\kappa-1)/\kappa$ and $C=\frac{2\kappa-1}{\kappa}(N_B+1)-1$ to find that
\begin{multline}
\sum_xp_X(x)H((\N^c)^{\otimes n}(\rho_x)) \\
\geq ng((\kappa-1)(\lambda N_S+1)+\kappa N_B)~.
\end{multline}
This concludes our proof. Together with the achievability of \eqref{eq:amplify-rate-Bob}--\eqref{eq:amplify-rate-Charlie}, we establish it as the  capacity region for the quantum broadcast amplifier channel, provided that the multi-mode version of Theorem~1 is true.
\end{proof}

\bigskip
Now let us consider the quantum-limited amplifier channel. Since the broadcast capacity region for Hadamard channels is single-letter \cite{wang2016capacities}, by setting $n=1$ and $N_B=0$ in the above proof, we 
establish the following:
\begin{corollary}
For a quantum-limited amplifier broadcast channel,
\eqref{eq:amplify-rate-Bob}--\eqref{eq:amplify-rate-Charlie} with $N_B=0$ is equal to the capacity region.
\end{corollary}

\subsection{Coherent-detection and large $\kappa$ limit}
\begin{figure}
	\centering
	\includegraphics[width=\columnwidth]{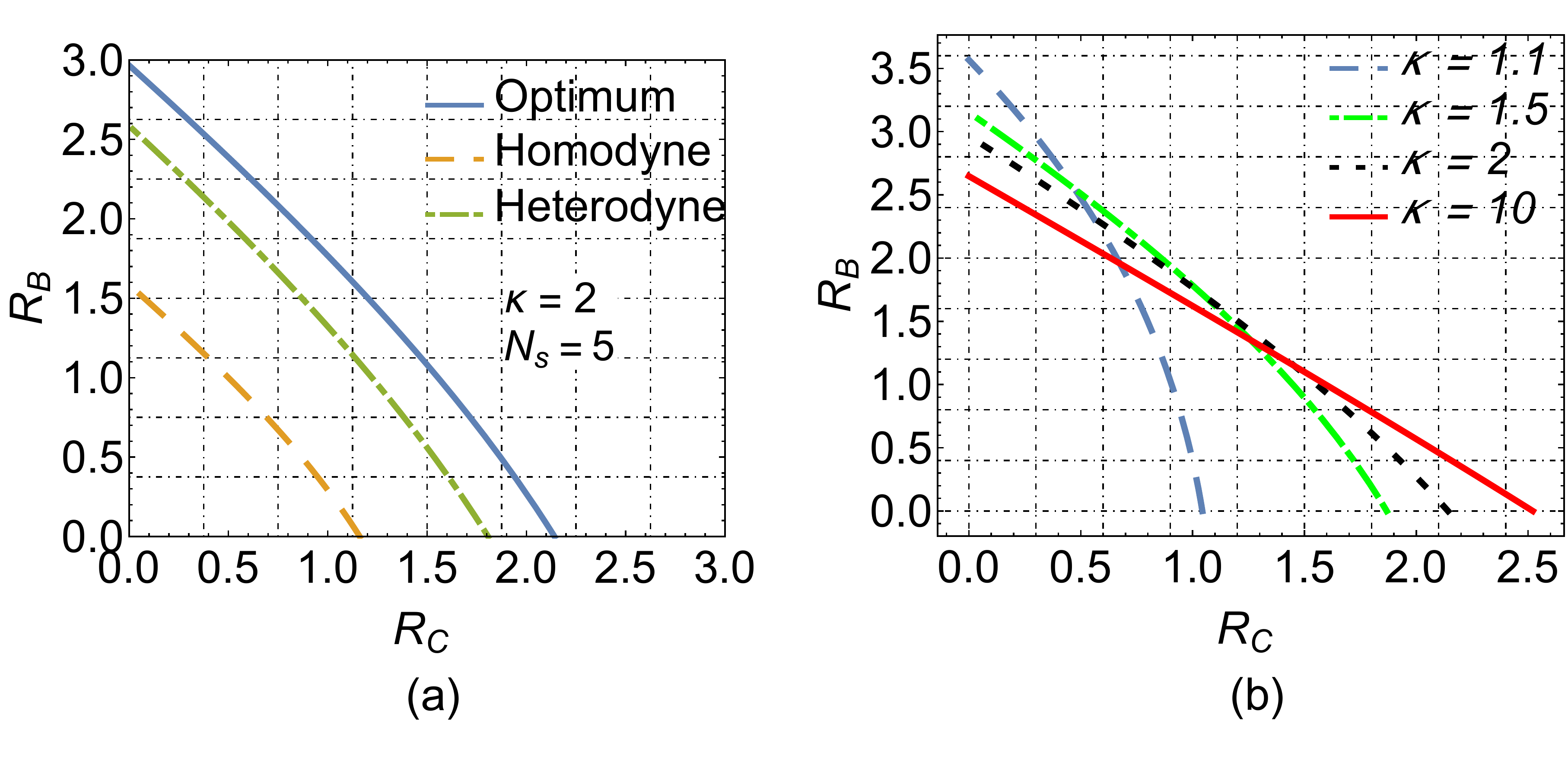}
	\caption{In (a) we consider a quantum-limited broadcast amplifier channel with
$N_S=5$ and $N_B=0$. We compare the capacity region obtained by homodyne detection (\eqref{eq:broadcast-coherent-detection} with $\xi=1/2$), heterodyne detection (\eqref{eq:broadcast-coherent-detection} with $\xi=1$) and the optimal measurement (\eqref{eq:amplify-rate-Bob}--\eqref{eq:amplify-rate-Charlie} with $N_B=0$).  In (b) we plot the large $\kappa$ limit of the rate region. At $\kappa=10$, it is indistinguishable with the limit in \eqref{eq:broadcast-region-limit}.
	\label{fig:broadcast-trade-offs}}
\end{figure}

To evaluate the performance of the capacity region given by 
\eqref{eq:amplify-rate-Bob} and \eqref{eq:amplify-rate-Charlie} with $N_B=0$, we compare it with what can be achieved by conventional, coherent-detection 
strategies \cite{yuen1980optical,caves1994quantum}.
When Alice inputs a coherent state $\ket{\alpha}$, Bob receives a displaced thermal state $D(\sqrt{\kappa}\alpha)
\rho^{\operatorname{th}}_{\bar{\kappa}}
D^\dagger(\sqrt{\kappa}\alpha)$, where $D(\sqrt{\kappa}\alpha)$ denotes a displacement operator and $\rho^{\operatorname{th}}_{\bar{\kappa}}$ the density operator corresponding to a thermal state with mean photon number $\bar{\kappa}$ \cite{GK04}. When Bob employs homodyne or heterodyne detection \cite{GK04}, his measurement outcomes have particular Gaussian distributions, and similarly for Charlie. The quantum broadcast channel then reduces to a classical Gaussian channel with additive noise \cite{cover2012elements}. Using known results for  classical Gaussian broadcast channels
\cite{cover2012elements,caves1994quantum,el2010lecture}, we find that coherent-detection strategies lead to the following capacity regions:
\begin{equation}
\label{eq:broadcast-coherent-detection}
\begin{aligned}
R_B &\leq \xi\log_2\!\left( 1+
\frac{\lambda\kappa N_S}{\xi(\xi +\bar{\kappa})}\right)~,\\
R_C &\leq \xi\log_2\!\left( 1+
\frac{(1-\lambda)\bar{\kappa} N_S}{
\xi(\xi +\bar{\kappa})+\lambda \bar{\kappa} N_S}
\right)~,
\end{aligned}
\end{equation}
where $\xi=1/2$  for homodyne
detection and $\xi = 1$ for
heterodyne detection. See Appendix \ref{sec:coherent-detection} for a detailed derivation.

In Figure~\ref{fig:broadcast-trade-offs}(a), we compare these strategies with the optimal strategy for a quantum-limited amplifier with $\kappa=2$ and  $N_S=5$. As we can see, the capacity region we find in \eqref{eq:amplify-rate-Bob} and \eqref{eq:amplify-rate-Charlie} outperforms both coherent detection schemes. For relatively high mean photon number, heterodyne detection outperforms homodyne detection as expected from prior results \cite{guha2007classical}. 

Notice that in the first equation of (\ref{eq:broadcast-coherent-detection}), the amplifier gain $\kappa$ happens to cancel out in the case of heterodyne detection ($\xi=1$). This indicates that amplifying will both boost and hurt the transmission rate, so that there should exist a `balanced point'. Actually, if we consider the large $\kappa$ limit, \eqref{eq:amplify-rate-Bob} and \eqref{eq:amplify-rate-Charlie}
reduce to a gain-independent linear trade-off:
\begin{align}
\label{eq:broadcast-region-limit}
R_B +R_C \leq \log_2(N_S/[N_B+1]+1)~.
\end{align}
Physically, although a large amplifier gain will amplify the input energy power and thus potentially increase the capacity, it is balanced out by the increasing noise generated from amplifying the vacuum, manifested by the negative terms in \eqref{eq:amplify-rate-Bob} and \eqref{eq:amplify-rate-Charlie}. With mean photon numbers $N_S=5$
and $N_B = 0$, the maximum classical capacity of Bob and Charlie converges to around $\log_2(6) \approx 2.58$ bits per channel use.
In Figure~\ref{fig:broadcast-trade-offs}(b),
we plot the rate region for amplifier gain $\kappa$
increasing 
from $1.1$ to $10$. The capacity region converges to \eqref{eq:broadcast-region-limit} very quickly. The maximum capacities for both receivers approach around 2.6 bits per channel use, as expected from the reasoning above. 

\section{Trading public and private resources}
\label{sec:public-private}

Here we briefly argue that we obtain the private dynamic capacity region \cite{WH10} of quantum-limited amplifier channels.
The techniques for establishing this result are similar to those from previous sections, so we merely state the result rather than going through all the details.

The information-theoretic task is similar to the triple trade-off discussed previously, but the resources involved are different.
Here we are concerned with the transmission (or consumption) of public classical bits, private classical bits, and secret key along with the consumption of many independent uses of a quantum-limited amplifier channel. The communication  trade-off is characterized by \textit{rate triples} $(R,P,S)$, where $R$ is the net rate of public classical communication, $P$ is the net rate of private classical communication, and $S$ is the net rate of secret key generation. 

Since the quantum-limited amplifier channel is a Hadamard channel, the private dynamic capacity region
of a quantum channel $\N$
is given by the union of regions of the following form \cite{WH10}:
\begin{align}
R+P    & \leq H(  \mathcal{N}(  \rho)  )
-\sum_{x,y}p_X(  x)p_{Y|X}(y|x)    H(  \N(\psi
_{x,y}))  ,\nonumber\\
P+S    & \leq\sum_{x}p_X(  x)  \left[  H(  \mathcal{N}(
\rho_{x})  )  -H(  \mathcal{N}^{c}(  \rho_{x})
)  \right]  ,\nonumber\\
R+P+S    & \leq H(  \mathcal{N}(  \rho)  )  -\sum
_{x}p_X(  x)  H(  \mathcal{N}^{c}(  \rho_{x})
),
\end{align}
where the union is with respect to all possible
pure-state input ensembles
$\brac{ p_X(x)p_{Y|X}(y|x), \psi_{x,y}}$,
\begin{align}
\rho_x & \equiv \sum_y p_{Y|X}(y|x) \psi_{x,y},\\
\rho & \equiv \sum_x p_X(x) \rho_x,
\end{align}
and $\N^c$ is a complementary channel of $\N$. To give an upper bound on the single-letter
private dynamic
capacity region of the quantum-limited amplifier channel, we need to show that for all input ensembles
$\brac{ p_X(x)p_{Y|X}(y|x), \psi_{x,y}}$, there exists a $\lambda \in [0,1]$ such that 
the following four inequalities hold
\begin{align}
H(\N(\rho))\leq g(\kappa N_S +\bar{\kappa})
~,\\
\sum_x\sum_y p_X(x)p_{Y|X}(y|x)H(\N(\psi_{x,y})) \geq g(\bar{\kappa})~,\\
\sum_x p_X(x)H(\N(\rho_x)) \leq
g(\kappa\lambda N_S +\bar{\kappa})~,\\
\sum_x p_X(x)H(\N^{c}(\rho_x)) \geq
g(\bar{\kappa}(\lambda N_S+1))~. 
\end{align}
We can establish these bounds using methods from the previous sections. Thus, we find that the private dynamic capacity region of the quantum-limited amplifier
channel is as follows:
\begin{align}
R+P    & \leq g(\kappa N_S +\bar{\kappa})
-g(\bar{\kappa})  ,\label{eq:PDC-1}\\
P+S    & \leq g(\kappa\lambda N_S +\bar{\kappa})  -g(\bar{\kappa}(\lambda N_S+1))    ,\\
R+P+S    & \leq g(\kappa N_S +\bar{\kappa})  -g(\bar{\kappa}(\lambda N_S+1)) 
).\label{eq:PDC-3}
\end{align}
This rate region is achievable as well, as shown in
\cite{wilde2012information,wilde2012quantum}, and so the union of \eqref{eq:PDC-1}--\eqref{eq:PDC-3} with respect
to $\lambda\in[0,1]$ is equal to the private dynamic capacity region.

\section{Discussion}
\label{sec:discussion}

Theorem~\ref{th:min-out-ent-conj} from  \cite{PTG16} plays an important role in our proof of the  capacity regions for the information trade-off and quantum broadcast settings.
For a long time now, thermal states have been conjectured to minimize the output entropy for pure-loss channels with an input entropy constraint \cite{guha2007classical}. The authors of  \cite{PTG16}
established this result
for all single-mode phase-insensitive bosonic Gaussian channels,
 going well beyond the original conjecture and including it as a special case. The special case for $H_0=0$ was proved for all multi-mode phase-insensitive 
Gaussian channels \cite{giovannetti2014ultimate,Giovannetti2014}. 
After that, de Palma \textit{et al.}~first 
reduced the optimizer problem to the set of all possible passive states \cite{de2015passive} using the technique of majorization \cite{arnold2011inequalities,giovannetti2004minimum,mari2014quantum} and subsequently proved the conjecture for single-mode pure-loss channels \cite{de2016gaussian}. The multi-mode generalization of the results in Ref.~\cite{PTG16}, which would determine capacity regions for pure-loss channels \cite{guha2007classical,wilde2012information}, is still unsolved.

The strongest conjecture proposed so far is the Entropy Photon number Inequality (EPnI) \cite{guha2008entropy} which takes on a role analogous to Shannon's entropy power inequality \cite{shannon1948bell}. The truth of the EPnI subsumes all minimum output entropy conjectures. Although the EPnI has not been proved yet, a different quantum analog of EPI, quantum EPI (qEPI) has been proved recently for a multi-mode lossy channel
\cite{de2014generalization}. Although the qEPI does not imply the truth of the EPnI, the lower bounds given by the two inequalities are extremely close for a large range of parameters \cite{de2014generalization}. This fact strongly suggests the truth of the multi-mode EPnI. We give an upper bound in Appendix~\ref{sec:loss-channel-EPI} for the capacity region of information trade-off over the pure-loss channel by using the qEPI. This represents the first application of the qEPI to the information trade-off problem. The bound given by the qEPI is extremely close to the upper bound, if we assume the multi-mode minimum output entropy conjecture is true. Therefore, it is safe to say that the achievable rate region found in Ref.~\cite{wilde2012information} is the optimal capacity region for all practical purposes. 

\section{Conclusion}
\label{sec:conclusion}

We have determined the capacity region for trade-off coding over quantum-limited amplifier channels and have shown that it can significantly outperform a
time-sharing strategy. We also find that with increasing amplifier gain, the capacity region is shrinking, due to amplification noise from the vacuum. Going beyond the point-to-point setup, we have also determined the classical capacity region for broadcast communication over a quantum-limited amplifier channel, which outperforms the communication rate achieved using conventional coherent detection. The capacity region converges to a linear trade-off form, with the
same maximum rate for both receivers, when the amplifier gain $\kappa$ is large.

One unsolved problem is to determine the trade-off capacity for a pure-loss channel \cite{wilde2012information,wilde2012quantum} and the broadcast capacity for a thermal-noise channel \cite{guha2007classical,guha2008multiple}, which  require a multi-mode version of Theorem~\ref{th:min-out-ent-conj}.
Recall that Appendix~\ref{sec:loss-channel-EPI} shows how it is possible to obtain a good bound for the trade-off capacity region by employing the quantum entropy power inequality \cite{de2014generalization}, but it is likely possible to improve this bound.
On the other hand, the techniques we used in our converse proofs in this work are not applicable to quantum channels which are not degradable. This includes lossy and amplifier channels with thermal noise and pure-loss channels with transmissivity smaller than one-half. We leave the problem of determining the capacities for those channels as future work.

\acknowledgments
We are grateful to
G.~De Palma,
R.~Garcia-Patron,
S.~Guha, 
E.~Martin-Martinez,
and M.~Takeoka for discussions.
HQ acknowledges support from AFOSR, ARO, and NSF.
MMW acknowledges support from NSF Award No.~CCF-1350397.

\appendix

\section{Two properties of $g(x)$}

\label{sec:property}
We  first recall a property of the function
\begin{equation}
g(x) = (x+1)\log_2 (x+1)-x\log_2 x,
\end{equation}
which is helpful for our converse proofs. Recall that $g(x)$ is equal to the entropy of a thermal state with mean photon number $x$.

\begin{theorem}[Theorem A.3 of \cite{guha2008multiple}]
	\label{thm:guha-A-3}
	Given $q \in [0,1]$ a probability distribution $p_X(x)$ and non-negative real numbers $\brac{y_x: 1\leq x \leq n}$, if
	\begin{align}
	\sum_{x=1}^np_X(x)g(y_x) = g(y_0)~,
	\end{align}
	then the following inequality holds
	for $C \geq 0$:
	\begin{align}
	\sum_{x=1}^np_X(x)g(q y_x +C)\geq g(q y_0 +C)~.
	\end{align}
	\end{theorem}
	
	As mentioned above, the above inequality is
	Theorem~A.3 in Appendix~C of \cite{guha2008multiple}. Observe
	that
	Ref.~\cite{guha2008multiple}
	proved the inequality for $p_X(x)$ set to the uniform distribution. However, the argument there only relies on concavity of $g(x)$ and thus applies to an arbitrary distribution, as discussed later in \cite{wilde2012quantum}.
	
	Due to the requirement that $q \in [0,1]$, Theorem~\ref{thm:guha-A-3} is not useful for the quantum amplifier channel given that its amplifier gain $\kappa > 1$. To resolve this problem, we prove another property of $g(x)$:
	
	\begin{theorem}\label{thm:second-prop-gx}
		Given $q \in (1,+\infty)$, a probability distribution $p_X(x)$ and non-negative real numbers $\brac{y_x: 1\leq x \leq n}$, if
		\begin{align}
		\sum_{x=1}^n p_X(x)g(y_x) = g(y_0),
		\end{align}
		then 
		\begin{align}
		\sum_{x=1}^n p_X(x)g(q y_x +q-1)\geq g(q y_0 +q-1)~.
		\end{align}
		\end{theorem}
		
		\begin{proof}
			The original proof of
			Theorem~\ref{thm:guha-A-3} depends on the following inequality:
			\begin{multline}
			\log_2\!\left(1+\frac{1}{qx+C}\right)(qx+C)(1+qx+C)\\
			\geq
			\log_2\!\left(1+\frac{1}{x}\right)qx(1+x)~, \label{eq:related-ineq}
			\end{multline}
			which holds for $q \in [0,1]$,
			$x \geq 0$, and $C\geq 0$. When considering $q>1$, the above inequality does not generally hold. However, we can prove that it is true for $C = q-1$.
			Substituting $C = q-1$ in \eqref{eq:related-ineq}, we need to show that
			\begin{align}
			(q(1+x)-1)\log_2\frac{q(1+x)}{q(1+x)-1}\geq x\log_2\frac{1+x}{x}~.
			\end{align}
			Defining $h(x) = x\ln\frac{x+1}{x}$, then we can see that the above inequality is equivalent to the following one:
			\begin{align}
			h(q(1+x)-1)\geq h(x)~.
			\end{align} 
			But 
			\begin{equation}
			\lim_{x\rightarrow 0}h(x) = \lim_{t\rightarrow +\infty}\frac{\ln(1+t)}{t}
			=\lim_{t\rightarrow +\infty} \frac{1}{1+t} = 0~,
			\end{equation}	
			by L'Hospital's rule, 
			and 
			\begin{align}
			h'(x) = \ln\!\left(1+\frac{1}{x}\right) -\frac{1}{1+x}~.
			\end{align}
			Since $h'(0) = +\infty$, $h'(+\infty)= 0$, and $h''(x) = -1/x(1+x)^2$, we have
			\begin{align}
			h'(x)\geq 0~,
			\end{align}
			for $x \geq 0$, and the 
			function $h(x)$ is non-negative and monotonically increasing for non-negative $x$. Now since
			$q(1+x)-1-x = (1+x)(q-1)\geq 0$, we find that
			\begin{align}
			h(q(1+x)-1)\geq h(x)~.
			\end{align}
			This concludes the proof.
			\end{proof}

\section{Coherent-detection schemes}

\label{sec:coherent-detection}

Although we have shown that \eqref{eq:amplify-rate-Bob}--\eqref{eq:amplify-rate-Charlie} is achievable by using coherent-state encoding with a Gaussian distribution, implicitly we have also assumed that it is achieved by some fully quantum measurement scheme. If the two receivers use classical coherent detection instead, the problem reduces to a classical broadcast channel with Gaussian additive noise. We expect such schemes to be outperformed by those achieved with a fully quantum measurement.

One way to calculate the capacity region of the classical degradable broadcast channel is to use the formula from \cite{yard2011quantum} with the same distribution as in \eqref{eq:input-distribution}. Another easier way is to first calculate the capacity of each classical channel to Bob and Charlie. Since each channel is Gaussian with additive noise, 
each capacity should have the following form: 
\begin{align}
C\equiv C(\mathtt{snr})=\frac{1}{2}\log_2(1+\mathtt{snr_{B/C}})~,
\end{align}
where $\mathtt{snr}_{B/C}$ is the signal-to-noise ratio of the channel $A\rightarrow B/C$. Then we can use known results \cite{cover2012elements,el2010lecture} to directly get the capacity region for broadcast channel,
\begin{align}
\label{eq:classical-capacity-region}
\nonumber
R_B &\leq C(\lambda \mathtt{snr}_B),\\
R_C & \leq C\!\left(\frac{(1-\lambda)\mathtt{snr}_C}{\lambda \mathtt{snr}_C+1}\right)~.
\end{align}
For Bob, the channel could be modeled by the following transformation:
\begin{align}
B = \sqrt{\kappa}A +Z~.
\end{align}
If homodyne detection is employed, Bob is measuring one of the quadratures and $B$, $A$, and $Z$ are scalar Gaussian random variables. The noise $Z$ has distribution $Z\sim N(0,\frac{1}{4}+\frac{1}{2}\bar{\kappa})$, where the variance comes from both the vacuum itself and the thermal noise generated from the vacuum \cite{Shap08}. The capacity of the classical Gaussian channel is achieved by input with distribution $A\sim N(0,N_S)$, and therefore we have 
\begin{align}
\mathtt{snr}_B = \frac{\kappa N_S}{\frac{1}{4}(1+2\bar{\kappa})}.
\end{align}
When heterodyne detection is used, $B$, $A$, and $Z$ are complex Gaussian random variables. The real part of the noise has distribution $\Re(Z) \sim N(0,\frac{1}{2}+\frac{1}{2}\bar{\kappa})$ and the same for the imaginary part \cite{Shap08}. The optimal input distribution for each part is
$\Re(A) \sim N(0,N_S/2)$ and $\mathrm{Im}(A) \sim N(0,N_S/2)$ since the total input power is $N_S$. Thus for heterodyne detection we have 
\begin{align}
\mathtt{snr}_B = \frac{\kappa N_S}{(1+\bar{\kappa})}=N_S.
\end{align}
Notice that we need to multiply the capacity formula by a factor of two, to take into account the contribution from each part of the complex variable. The channel to Charlie is modeled by
\begin{align}
C = \sqrt{\bar{\kappa}}A+Z~,
\end{align}
and all the analysis above for Bob still holds. We can write the capacity of each classical channel achieved by coherent detection in a unified way as
\begin{align}
C_{A\rightarrow B} &= \xi\log_2\!\left( 1+\frac{\kappa N_S}{\xi(\xi +\bar{\kappa})}\right)~,\\
C_{A\rightarrow C} &= \xi\log_2\!\left( 1+\frac{\bar{\kappa} N_S}{\xi(\xi +\bar{\kappa})}\right)~,
\end{align}
where $\xi = \frac{1}{2}$ for homodyne detection and $\xi = 1$ for heterodyne detection.
 
Now using \eqref{eq:classical-capacity-region}, we find the capacity region of coherent detection:
\begin{align}
R_B &\leq \xi\log_2\!\left( 1+\frac{\lambda\kappa N_S}{\xi(\xi +\bar{\kappa})}\right)~,\\
R_C &\leq \xi\log_2\!\left( 1+\frac{(1-\lambda)\bar{\kappa} N_S}{\xi(\xi +\bar{\kappa})+\lambda \bar{\kappa} N_S}\right)~,
\end{align}
thus giving \eqref{eq:broadcast-coherent-detection}.

\section{Upper bound for trade-off capacity region of the pure-loss channel given by the qEPI}

\label{sec:loss-channel-EPI}

The capacity region for the information trade-off over a pure-loss channel has been given in \cite{wilde2012information}, provided that a multi-mode minimum output-entropy conjecture is true.
Although the multi-mode conjecture has not been proved yet, the recently proved quantum EPI (qEPI) can give a good upper bound \cite{de2014generalization,konig2014entropy}, holding for $\eta \in[1/2, 1]$. The qEPI is a direct translation of the classical EPI and is as follows:
\begin{align}
\label{eq:qEPI}
2^{H(\rho_B)/n} \geq \lambda_A 2^{H(\rho_A)/n} + \lambda_E 2^{H(\rho_E)/n}~,
\end{align} 
where $\rho_A$ is the input of one beamsplitter port,
$\rho_E$ is the input of the other beamsplitter port,
$\rho_B$ is the output of one port,
and $\lambda_A = \eta, \lambda_E = 1-\eta $ for a pure-loss channel with transmissivity $\eta$.

 When we consider the case when the environment is in the vacuum state, we have $H(\rho_E)=0$ and the following bound holds
\begin{align}
\label{eq:qEPI-vacuum}
H(\rho_B)/n \geq \log_2(\eta 2^{H(\rho_A)/n} +1-\eta)~.
\end{align}
We will use this lower bound in what follows.

Recall the development in Eqs.~(60)--(76) of \cite{wilde2012quantum}. Picking up from there, we have that
\begin{align}
\sum_xp_X(x)H(\rho_x)&= ng(\lambda' N_S)~,\\
\sum_xp_X(x)H(\N(\rho_x))& = ng(\lambda\eta N_S)~,
\end{align}
where $\lambda',\lambda\in[0,1]$.
Now instead of invoking the minimum output-entropy conjecture for a pure-loss channel, we use the lower bound given by the multi-mode qEPI in~\eqref{eq:qEPI-vacuum}:
\begin{align}
\nonumber
g(\lambda\eta N_S)& \geq \sum_x p_X(x)\log_2(\eta 2^{g(\lambda_x' N_{S,x})}+1-\eta)~,\\
& \geq \log_2(\eta 2^{g(\lambda' N_S)}+1-\eta)~.
\end{align}
The last inequality follows from the fact that $f(x)=\log_2(\eta 2^{x}+1-\eta)$ is convex, and we have also used  the equality $\sum_xp_X(x)g(\lambda'_x N_{S,x})=g(\lambda'N_S)$. Rewriting this, we find that
\begin{align}
& \!\!\!\!\sum_x p_X(x) H(\rho_x) \nonumber\\
& =ng(\lambda'N_S)\\
& \leq n\log_2\left[\frac{1}{\eta}\left(2^{g(\lambda\eta N_S)}-(1-\eta\right)\right]~,
\end{align}
which replaces Eq.~(60) in Ref.~\cite{wilde2012quantum}.

 The lower bound given in Eq.~(63) of
 \cite{wilde2012quantum} will be replaced by a new lower bound found by invoking the qEPI. Using \eqref{eq:qEPI-vacuum} for a pure-loss channel with $\eta'=(1-\eta)/\eta$, we find that
\begin{align}
\nonumber
& \!\!\!\! \sum_x p_X(x) H(\N^c(\rho_x)) \\
&\geq \sum_x p_X(x) n\log_2(\eta'2^{g(\eta \lambda_x N_{S,x})}+1-\eta')\\
&\geq n\log_2(\eta' 2^{\sum_xp(x)g(\eta \lambda_x N_{S,x} )}+1-\eta')~\\
&= n\log_2\!\left(\frac{1-\eta}{\eta}2^{g(\eta \lambda N_S +1-\eta )}+\frac{2\eta-1}{\eta}\right)~.
\end{align}
The two inequalities follow by invoking the qEPI and convexity of $f(x)$ as defined and used previously.
In the last step, we have used $\sum_x p_X(x) g(\eta\lambda_x N_{S,x}) = g(\lambda\eta N_S)$.

In summary, an upper bound for the trade-off capacity region of the pure-loss channel, derived from the qEPI,  follows from the inequalities below:
\begin{align}
\frac{1}{n}H(\N^{\otimes n}(\rho))&\leq g(\eta N_S)
~,\nonumber\\
\frac{1}{n}\sum_x p_X(x)H(\rho_x)&\leq \log_2\left[\frac{1}{\eta}\left(2^{g(\lambda\eta N_S)}-(1-\eta)\right)\right]~,\nonumber\\
\frac{1}{n}\sum_x p_X(x)H(\N^{\otimes n}(\rho_x)) &\leq
g(\eta\lambda N_S )~,\nonumber\\
\frac{1}{n}\sum_x p_X(x)H((\N^{c})^{\otimes n}(\rho_x)) &\geq
\log_2\left[\frac{1-\eta}{\eta}2^{g( \lambda\eta N_S )}+\frac{2\eta-1}{\eta}\right]~.
\end{align}

For the broadcast capacity region of a pure-loss channel, the upper bound is given in Sec.~IV.C of \cite{de2014generalization}. 

\bibliography{amplifierPRL}
\bibliographystyle{unsrt}

\end{document}